\documentclass[conference]{ieeetran}
\IEEEoverridecommandlockouts

\usepackage[utf8]{inputenc}
\usepackage{amsfonts}
\usepackage{amsmath}
\usepackage{amsthm}
\usepackage{xspace}
\usepackage{xcolor}
\usepackage{tabularx}
\usepackage{dutchcal} 
\usepackage{graphicx}
\usepackage{caption}
\usepackage{subcaption}
\usepackage{comment} 
\usepackage{mathtools}
\usepackage[mathscr]{euscript}
\usepackage{hyperref}
\usepackage{cleveref}
\usepackage[normalem]{ulem}
\usepackage{comment}
\usepackage[misc]{ifsym} 
\usepackage{bbm}

\def\BibTeX{{\rm B\kern-.05em{\sc i\kern-.025em b}\kern-.08em
    T\kern-.1667em\lower.7ex\hbox{E}\kern-.125emX}}

\newcommand\edit[1]{\textcolor{black}{#1}}
\newcommand\toremove[1]{\textcolor{red}{}}
\newcommand\editAllerton[1]{\textcolor{black}{#1}}

\newcommand\MC[1]{$\clubsuit$\footnote{MC: #1}}

\newcommand\IR[1]{$\dag$\footnote{IR: #1}}

\newtheorem{theorem}{Theorem}
\newtheorem{lemma}{Lemma}

\newtheorem{remark}{Remark}
\newtheorem{definition}{Definition} 
\newtheorem{problem}{Problem} 

\newcommand{\sys}{\ensuremath{\mathsf{M}_{\text{sys}}}\xspace}

\newcommand{\mdl}{\ensuremath{\mathsf{M}_{\text{abs}}}\xspace}

\newcommand{\mcp}{\ensuremath{\mathsf{M}_{\text{cpl}}}\xspace}

\newcommand{\mper}{\ensuremath{\textsf{M}_{\text{per}}}\xspace}

\newcommand{\mdp}{\ensuremath{\mathsf{M}^{\text{mdp}}}\xspace}

\newcommand{\imdp}{\ensuremath{\mathsf{M}^{\text{imdp}}}\xspace}

\newcommand{\ourpath}{\ensuremath{\pi}\xspace}

\newcommand{\pset}{\ensuremath{\Pi}\xspace}

\newcommand{\timestep}{\ensuremath{\tau}\xspace}

\newcommand{\oursafeprop}{\ensuremath{\psi_{nocol}}\xspace}

\newcommand{\carlen}{\ensuremath{L}\xspace}

\newcommand{\modelOurs}{\ensuremath{\mathsf{M}_{ours}}\xspace}

\newcommand{\modelNoDeriv}{\ensuremath{\mathsf{M}_{oursNPE}}\xspace}

\newcommand{\modelNoCI}{\ensuremath{\mathsf{M}_{noCI}}\xspace}

\newcommand{\modelGTPer}{\ensuremath{\mathsf{M}_{GTPer}}\xspace}

\newcommand{\modelLogRegCI}{\ensuremath{\mathsf{M}_{logRegCI}}\xspace}

\newcommand{\probEnlargeFactor}{\ensuremath{w_{PE}}}

\begin{document}

\title{Conservative Perception Models\\ for Probabilistic Verification*
\thanks{DISTRIBUTION STATEMENT A. Approved for public release. Distribution is unlimited.}
}

\makeatletter
\newcommand{\linebreakand}{%
  \end{@IEEEauthorhalign}
  \hfill\mbox{}\par
\mbox{}\hfill
\begin{@IEEEauthorhalign}
}
    
\author{\IEEEauthorblockN{1\textsuperscript{st} Matthew Cleaveland}
\IEEEauthorblockA{\textit{MIT Lincoln Laboratory} \\
Lexington, MA, USA \\
matthew.cleaveland@ll.mit.edu}
\and
\IEEEauthorblockN{2\textsuperscript{nd} Pengyuan Lu}
\IEEEauthorblockA{
\textit{University of Pennsylvania}\\
Philadelphia, PA, USA \\
pelu@seas.upenn.edu}
\and
\IEEEauthorblockN{3\textsuperscript{rd} Oleg Sokolsky}
\IEEEauthorblockA{
\textit{University of Pennsylvania}\\
Philadelphia, PA, USA \\
sokolsky@seas.upenn.edu}
\linebreakand
\IEEEauthorblockN{4\textsuperscript{th} Insup Lee}
\IEEEauthorblockA{
\textit{University of Pennsylvania}\\
Philadelphia, PA, USA \\
lee@seas.upenn.edu}
\and
\IEEEauthorblockN{5\textsuperscript{th} Ivan Ruchkin}
\IEEEauthorblockA{
\textit{University of Florida}\\
Gainesville, FL, USA \\
iruchkin@ece.ufl.edu}
}

\maketitle




\begin{abstract}
    Verifying the behaviors of autonomous systems with learned perception components is a challenging problem due to the complexity of the perception and the uncertainty of operating environments. Probabilistic model checking is a powerful tool for providing guarantees on stochastic models of systems. However, constructing model-checkable models of black-box perception components for system-level mathematical guarantees has been an enduring challenge. In this paper, we propose a method for constructing provably conservative Interval Markov Decision Process (IMDP) models of closed-loop systems with perception components. We prove that our technique results in conservative abstractions with \editAllerton{a user-specified probability}. We evaluate our approach in an automatic braking case study using both a synthetic perception component and the object detector YOLO11 in the CARLA driving simulator. 
\end{abstract}


\section{Introduction}
\label{sec:intro}

\looseness=-1
Model checking is a widespread method for verifying the safety of autonomous systems~\cite{katoen_probabilistic_2016}. However, advances in machine learning have led to systems with complex learning-enabled components (LECs). These LECs tend to exhibit unpredictable, difficult-to-model behaviors~\cite{faria_non-determinism_2017,burton_making_2017,yampolskiy_predicting_2018}. The uncertainties of LECs and real-world environments are frequently handled with probabilistic model checking (PMC)~\cite{kwiatkowska_probabilistic_2022,katz_generating_2021}, \editAllerton{which can be used to compute} the probability that a safety property holds. PMC is particularly suitable for analyzing highly uncertain systems since it explicitly models both stochastic and non-deterministic behaviors.

Using PMC to verify closed-loop systems requires building a stochastic model of the system, usually as a probabilistic automaton (PA) or a Markov decision process (MDP). This model is often abstracted from a more detailed description, such as differential or difference equations. Then model checking tools, such as PRISM~\cite{kwiatkowska_prism_2011} or Storm~\cite{hensel_probabilistic_2022}, are used to analyze that abstraction, such as computing the probability of the car colliding with an obstacle. If the abstraction overapproximates the true underlying system (i.e., contains all of its possible behaviors), PMC provides safety guarantees about the true system. Such abstractions are referred to as \emph{conservative abstractions}. 

\looseness=-1
Alas, constructing conservative abstractions of realistic systems is difficult. Consider a closed-loop autonomous system consisting of a controller, sensor, perception, and physical dynamics. Discretizing and conservatively modeling the LEC-free control and dynamics components as \edit{an MDP} with non-determinism is well-understood (see, e.g., \cite{cleaveland_monotonic_2022}). However, constructing provably conservative \edit{MDP} abstractions of high-dimensional sensing (e.g., a camera) and LEC-based perception is an open question in the literature, due to the uncertainty of the environment (e.g., the weather conditions affecting the camera) and the complexity of the perception algorithms (e.g., a large object detection neural network). 

Two major uncertainties obstruct the creation of provably conservative probabilistic abstractions of sensing and perception LECs, even given a sizeable validation dataset. 
First, one must account for the inherent uncertainty in using data-driven methods to fit \edit{MDP} parameters, e.g., by considering the confidence intervals of the estimated parameters of the model. Second, it is difficult to account for the mismatch between discrete \edit{MDP} models and sensing/perception components that vary continuously with respect to the (continuous) state of the true system. Existing perception abstraction methods \edit{either only consider the first challenge~\cite{puasuareanu_closed_2023}, or} consider \edit{neither} of these challenges~\cite{badithela_evaluation_2023,cleaveland_conservative_2023}, instead just assuming that their perception abstractions are conservative.

This paper develops a data-driven approach to constructing provably conservative \edit{MDP} models of the sensing and perception components for verification of LEC-based autonomous systems. As system models, we employ \emph{interval Markov decision processes (IMDPs)} --- probabilistic models that allow for transition probabilities to take any value within an interval. Our approach designs specialized \textit{confidence intervals (CIs)} to account for the inherent uncertainty in fitting a model using data (the first uncertainty above) and \textit{logistic regression} to estimate how the continuous perception behaviors change within the discrete states of the abstract model (the second uncertainty above). We prove that our abstractions are conservative with high probability.

We evaluate our approach on an automatic emergency braking system (AEBS) for an autonomous car, both with a fully synthetic perception component and with a YOLO11 perception component in the self-driving car simulator CARLA \cite{dosovitskiy_carla_2017}. Our experiments show that our abstractions are indeed empirically conservative, unlike simpler baselines. We also investigate how different hyperparameters of our approach affect the resulting safety estimates of the AEBS.

The contributions of this paper are three-fold:
\begin{itemize}
    \item We introduce a method for constructing conservative models of learning-enabled perception components with high probability.
    \item We prove that our perception abstractions, when composed with conservative abstractions of controller/plant models, produce conservative system models.
    \item We empirically evaluate our abstractions on AEBS case studies using both synthetic and learning-based perception models.
\end{itemize}

The rest of the paper is organized as follows. We survey related work in \Cref{sec:related_work} and give the necessary formal background in \Cref{sec:background}. \editAllerton{\Cref{sec:problem} introduces our problem statement, and} \Cref{sec:approach} introduces our perception abstraction approach. We state our conservatism guarantee in \Cref{sec:guarantee}. We describe the results of our two case studies in \Cref{sec:caseStudies}. The paper ends with a conclusion in \Cref{sec:discussion}.

\section{Related Work}
\label{sec:related_work}



\subsection{Safe Perception}
\label{subsec:safe_perception}


Autonomous safety depends on the accuracy of perception. Therefore, researchers in safety-critical systems aim to improve vision quality. First, fail-safe perceptions can be constructed by improved training quality. For instance, training quality may be measured by adversarial robustness~\cite{madry_towards_2018,seshia_semantic_2020}, so that probabilistic guarantees in vision accuracy can be achieved~\cite{robey_probabilistically_2022}. After training, there are methods for verifying the robustness of perception components \cite{paterson_deep_2021}. An alternative approach is to assist the perception with online monitoring and intervention so that inaccurate predictions are replaced. Such approaches include detectors for unfamiliar inputs (which are more likely to be mis-predicted)~\cite{richter_safe_2017} and diagnostic graphs that online analyze the cause of mis-predictions through module dependency~\cite{antonante_monitoring_2022}. Notice that neither of the two approaches guarantees successful predictions on sets of inputs. Moreover, researchers have assisted the above two solutions methods to identify the vulnerable inputs to perception models~\cite{kim_programmatic_2020,delecki_how_2022}. Specifically, they inject manually designed disturbance into the perception's inputs to mimic different adversarial scenarios, to find the inputs with a high likelihood of failure. Still, such testing is only able to identify point-wise vulnerable inputs but does not analyze continuous input sets. Additionally, researchers have developed techniques for incorporating system-level safety requirements into perception loss functions~\cite{corso_risk_2022} and those same safety requirements into accuracy requirements of the perception models~\cite{katz_predicting_2023}, providing useful information in perception construction and fine-tuning.




\subsection{Model Checking}
\label{subsec:model_checking}

Model checking is an established approach to verify whether a system satisfies a given formal specification. Verification literature first focuses on deterministic overapproximate reachability, such as Verisig~\cite{ivanov_verisig_2021}, which proves specifications for neural networks with nonlinear activation functions. Recent research emphasizes probabilistic verification, where the system is represented as a stochastic process. These representations include discrete and continuous-time Markov chains (DMTCs and CMTCs), Markov decision processes (MDPs), probabilistic timed automata (PTAs)~\cite{kwiatkowska_prism_2011,lahijanian_formal_2015,zhao_probabilistic_2019,alasmari_quantitative_2022}, as well as interval Markov chains and Markov decision processes (IMC and IMDP)~\cite{zhao_interval_2020,badings_probabilities_2023} and MDPs with convex uncertainties \cite{puggelli_polynomial_2013}. Existing works in the literature have also investigated controller synthesis on IMCs and IMDPs \cite{jiang_safe_2022,mathiesen_accelerated_2024,wooding_impact_2024,jiang_local_2024,calinescu_controller_2024}. We build on this literature to verify a system model that provably accounts for black-box perception uncertainty.

\subsection{Perception Abstractions}
\label{subsec:perception_abstractions}

Another studied problem is how to abstractly represent the perception process. Many approaches using probabilistic models like Markov Decision Processes (MDPs) or similar to abstract perception components for model checking~\cite{puasuareanu_closed_2023,badithela_evaluation_2023,cleaveland_conservative_2023,peleska_stochastic_2023,mitra_formal_2025}. Other uses of perception models include run-time monitoring~\cite{ruchkin_confidence_2022,xie_mosaic_2023,antonante_task_2023} and falsification~\cite{dreossi_verifai_2019,xie_mosaic_2023}. Another strand of work constructs deterministic abstractions of perception components, such as set-based piece-wise linear functions~\cite{hsieh_verifying_2022} and computable error bounds~\cite{wang_bounding_2021}. Yet another line of research constructs conservative abstractions of full system behaviors using stochastic models like MDPs~\cite{ajeleye_data_2023} or Generalized Polynomial Chaos~\cite{joshi_generating_2023}. These approaches do not (and often cannot) prove that their model of a real-world object detector is conservative with respect to physical reality. 

 Alternatively, for specific vision systems, researchers have built detailed parameterized models (e.g., of a pinhole camera), primarily for road-following autonomous vehicle cameras~\cite{habeeb_verification_2023,joshi_generating_2023,puasuareanu_closed_2023,piazzoni_PEM_2024,habeeb_interval_2024} but also aircraft landing cameras~\cite{santa_cruz_nnlander-verif_2022} and racing car LiDARs~\cite{ivanov_case_2020}. However, such approaches are difficult to generalize to visually complex phenomena in a wide range of systems. 
 
 The most similar works to ours \edit{are}~\cite{badithela_evaluation_2023,puasuareanu_closed_2023}. \edit{In~\cite{badithela_evaluation_2023},} the authors model the perception component of an autonomous system as a distance-parameterized confusion matrix. They then fit the parameters of the model using perception data from the system. However, the authors do not account for either of the two uncertainties mentioned in the introduction. Therefore, their abstraction may not be conservative, which could lead to a false positive verification verdict. \edit{In~\cite{puasuareanu_closed_2023}, the authors use uncertainty quantification methods to compute confidence intervals for perception LECs. However, their approach does not consider the mismatch between discrete PA models and continuous sensing/perception components.}




\section{Background: Probabilistic Model Checking}
\label{sec:background}

\begin{definition}[MDP] \label{def:mdp}
    A \emph{Markov Decision Process} (MDP) is a tuple $\mdp = (S,s_0,\alpha,\delta,L)$, where $S$ is a finite set of states, $s_0 \in S$ is the initial state, $\alpha$ is an alphabet of action labels, $\delta: S \times \alpha \times S \rightarrow [0,1]$ such that $\forall s\in S, a\in \alpha, \; \sum_{s'\in S} \delta(s,a,s')=1$ is a probabilistic transition function, and $L: S \rightarrow 2^{AP}$ is a state labeling function with atomic propositions $AP$.
\end{definition}
If $\delta(s,a,s')=p$, then \mdp can make a transition in state $s$ with action $a$ to state $s'$ with probability $p$. A state $s$ is \textit{terminal} if $~\forall a\in \alpha: \delta(s,a,s)=1$. A \textit{path} in \mdp is a (potentially infinite) sequence of transitions $\ourpath=s_0 \xrightarrow{a_1} s_1 \xrightarrow{a_2} \hdots$ with $\delta(s_i,a_i,s_{i+1}) \neq 0$. A \textit{set of paths} is denoted as $\pset$. We use $\mdp(s)$ to denote the MDP $\mdp$ with initial state $s$.
\begin{remark}
    An analogous, and perhaps more common, way of defining an MDP is to have the transition relation be $(S,\alpha,Dist(S)) \in \delta$, where $(s,a,\mu)\in \delta$ means that the MDP can transition from state $s$  with action label $a$ and move based on distribution $\mu$ to state $s'$ with probability $\mu(s')$. We adopt the functional definition in Def.~\ref{def:mdp} of the transition relation $\delta$ because it parallels the definition of transition relations 
    below.
\end{remark}

\begin{definition}[IMDP] \label{def:imdp}
    An \emph{Interval Markov Decision Process} (IMDP) is a tuple $\imdp  = (S,s_0,\alpha,\Delta,L)$, where $S$ is a finite set of states, $s_0 \in S$ is the initial state, $\alpha$ is an alphabet of action labels, $\Delta: S \times \alpha \times S \rightarrow  \{ [b,c] \mid 0 < b \leq c \leq 1 \} \cup \{[0,0] \}$ is an interval transition probability function, and $L: S \rightarrow 2^{AP}$ is a state labeling function with atomic propositions $AP$. We further decompose $\Delta$ into its lower and upper components, denoted $\Delta^{l}$ and $\Delta^{u}$, respectively: $\Delta(s,a,s') = [\Delta^{l}(s,a,s'),\Delta^{u}(s,a,s')]$.
\end{definition}

\begin{remark}
    The requirement that $b > 0$ in the definition of $\Delta$ ensures that the graph structure of the IMDP stays the same for all the values of $\Delta(s,a,s')$. This is necessary for IMDP model checking algorithms on which we rely.
\end{remark}

Now we define the usual parallel composition for MDPs and IMDPs. In this paper, the perception component will be represented by an MDP/IMDP and the dynamics/control will be represented by an MDP. 

\begin{definition}[MDP composition]\label{def:mdpComp}
    The \emph{parallel composition} of MDPs $\mdp_1 = \{S_1,s_{0}^{1},\alpha_1,\delta_1,L_1\}$ and $\mdp_2 = \{S_2,s_{0}^{2},\alpha_2,\delta_2,L_2\}$ is given by the MDP $\mdp_1 || \mdp_2 = \{ S_1 \times S_2,(s_{0}^{1},s_{0}^{2}),\alpha_1 \cup \alpha_2, \delta,L \}$, where $L(s_1,s_2) = L_1(s_1)\cup L_2(s_2)$ and 
    \begin{equation}
        \begin{split}
            &\delta((s_1,s_2),a,(s_{1}',s_{2}')) \\
            &= \begin{cases}
        \delta_1(s_1,a,s_1')\cdot \delta_2(s_2,a,s_2') &\text{ if } a\in \alpha_1 \bigcap \alpha_2 \\ 
        \delta_1(s_1,a,s_1') &\text{ if } a\in \alpha_1 \setminus \alpha_2 \land s_2 = s_2' \\
        \delta_2(s_2,a,s_2') &\text{ if } a\in \alpha_2 \setminus \alpha_1 \land s_1 = s_1' \\
        0 &\text{ otherwise }
            \end{cases}
        \end{split}
    \end{equation}
\end{definition}

\begin{definition} [MDP-IMDP composition]\label{def:imdpComp}
    The \emph{parallel composition} of MDP $\mdp = \{S_1,s_{0}^{1},\alpha_1,\delta_1,L_1\}$ with IMDP $\imdp  = \{S_2,s_{0}^{2},\alpha_2,\Delta_2,L_2\}$ is given by the IMDP $\mdp||\imdp= \{ S_1 \times S_2, (s_{0}^{1},s_{0}^{2}),\alpha_1 \cup \alpha_2,\Delta,L\}$, where $L(s_1,s_2) = L_1(s_1)\cup L_2(s_2)$ and 
    \begin{equation}
    \begin{split}
    &\delta((s_1,s_2),a,(s_{1}',s_{2}'))\\
    &= \begin{cases}[\delta_1(s_1,a,s_1')\cdot \Delta^{l}_2(s_2,a,s_2'),\delta_1(s_1,a,s_1')\cdot \Delta^{u}_2(s_2,a,s_2')]\\
    \quad\quad\quad \text{ if } a\in \alpha_1 \bigcap \alpha_2 \\ 
        [\delta_1(s_1,a,s_1'),\delta_1(s_1,a,s_1')]\\ \quad\quad\quad \text{ if } a\in \alpha_1 \setminus \alpha_2 \land s_2 = s_2' \\
        [ \Delta^{l}_2(s_2,a,s_2'),\Delta^{u}_2(s_2,a,s_2')]\\
        \quad\quad\quad\text{ if } a\in \alpha_2 \setminus \alpha_1 \land s_1 = s_1' \\
        [0,0] \quad \text{ otherwise }
    \end{cases}
    \end{split}
\end{equation}
\end{definition}

We also need a way of determining if the stochastic behaviors of an MDP \mdp are contained within the set of stochastic behaviors defined by an IMDP \imdp. To do that, we start with the known \edit{general} notion of satisfaction relations~\cite{jonsson_specification_1991}.

\begin{definition}[MDP-IMDP satisfaction][Definition 2 in \cite{delahaye_decision_2011}]\label{def:consPerGeneral}
    MDP $\mdp = \{S_1,s_{0}^{1},\alpha_1,\delta_1,L_1\}$ \emph{implements}  IMDP $\imdp  = \{S_2,s_{0}^{2},\alpha_2,\Delta_2,L_2\}$ if there exists a \emph{satisfaction relation} $R \subseteq S_1\times S_2$ such that whenever $s_1 R s_2$, then:
    \begin{itemize}
        \item $L_1(s_1) = L_2(s_2)$
        \item There exists a \emph{correspondence function} $\varphi: S_1 \rightarrow (S_2 \rightarrow [0,1])$ such that 
        \begin{itemize}
            \item For all $s'_1 \in S_1$ \edit{and $a_1 \in \alpha_1$}, if $\delta_1(s_1,a_1,s'_1) > 0$ then $\varphi(s'_1)$ defines a distribution on $S_2$,
            \item \edit{For all $a_1 \in \alpha_1$, there exists $a_2 \in \alpha_2$} such that $\sum_{s'_1 \in S_1} \delta_1(s_1,a_1,s'_1) \varphi(s'_1)(s'_2)\in \Delta_2(s_2,a_2,s'_2)$ for all $s'_2 \in S_2$, and
            \item \editAllerton{For all $s'_1 \in S_1$ and $s'_2 \in S_2$,} if $\varphi(s'_1)(s'_2)> 0$, then $s'_1 R s'_2$
        \end{itemize}
    \end{itemize}
\end{definition}

Note that the correspondence function $\varphi$ maps a state from \mdp to a distribution of states in \imdp. It is needed for the general case, when there may not be a direct correspondence of states in \mdp to states in \imdp, hence the mapping to state distributions. However, in our case, we already know that correspondence: the states of both processes are the states of our autonomous system. So for brevity, we restate this definition when \mdp and \imdp share a common state space $S$ and the states in \mdp directly map to states in \imdp. 

\begin{definition}[State-matched MDP-IMDP satisfaction]\label{def:consPerOurs}
    For MDP $\mdp = \{S,s_{0}^{1},\alpha_1,\delta_1,L_1\}$ and IMDP $\imdp  = \{S,s_{0}^{2},\alpha_2,\delta_2,L_2\}$ with the same state set $S$,  \mdp \emph{implements} \imdp if $\forall s\in S$:
    \begin{itemize}
        \item $L_1(s) = L_2(s)$,
        \item \edit{For all $a_1\in \alpha_1$, there exists $a_2 \in \alpha_2$ such that} $ \delta_1(s,a_1,s') \in \Delta_2(s,a_2,s')$
    \end{itemize}
\end{definition}

To summarize, the transition probabilities of \mdp are contained within the transition probability intervals of \imdp. Def.~\ref{def:consPerOurs} is a special case of Def.~\ref{def:consPerGeneral} when \edit{$\varphi(s) =\mathbbm{1}_s$}.

Reasoning about MDPs and IMDPs also requires the notion of \textit{schedulers}, which resolve non-determinism to execute the models. For our purposes, a scheduler $\sigma$ maps each state of the MDP/IMDP to an available action label in that state. We use $\pset_{\imdp}^{\sigma}$ for the set of all paths through $\imdp$ when controlled by scheduler $\sigma$ and $Sch_{\imdp}$ for the set of all schedulers for $\imdp$. \edit{Similar notions apply to MDPs as well.} 

This paper is concerned with the probabilities of safety properties in the aforementioned stochastic models, which we state using \emph{Safety} Linear Temporal Logic (LTL) formulas over state labels \cite{kupferman_model_2001}.

\begin{definition}[Probability of LTL formula]\label{def:prob}
For LTL formula $\psi$, IMDP $\imdp$ \edit{(or MDP $\mdp$)}, and scheduler $\sigma \in Sch_{\imdp}$ \edit{(or $\sigma \in Sch_{\mdp}$)}, the \emph{\edit{upper and lower probabilities} of $\psi$ holding on $\imdp$} are:
\begin{equation}
    \begin{split}
    Pr_{\imdp}^{\sigma,u}(\psi) &\coloneqq Pr_{\imdp}^{\sigma,u}\{ \pi \in \pset_{\imdp}^{\sigma}~|~\pi \models \psi \}\\
    Pr_{\imdp}^{\sigma,l}(\psi) &\coloneqq Pr_{\imdp}^{\sigma,l}\{ \pi \in \pset_{\imdp}^{\sigma}~|~\pi \models \psi \},
    \end{split}
\end{equation}
and \editAllerton{the probability of $\psi$ holding on  $\mdp$ is}
\begin{equation}
    Pr_{\mdp}^{\sigma}(\psi) \coloneqq Pr_{\mdp}^{\sigma}\{ \pi \in \pset_{\mdp}^{\sigma}~|~\pi \models \psi \},
\end{equation}
where $\pi \models \psi$ indicates that the path $\pi$ satisfies $\psi$ in the standard LTL semantics~\cite{pnueli_temporal_1977}. We limit our scope to LTL safety properties, which are LTL specifications that can be falsified by a finite trace.
\end{definition}

Probabilistically verifying an LTL formula $\psi$ against \imdp requires checking that the probability of satisfying $\psi$ meets the desired probability bounds for all schedulers. This involves computing the minimum or maximum probability of satisfying $\psi$ over all schedulers:\edit{
\begin{equation}
\begin{split}
    Pr_{\imdp}^{min}(\psi) &\coloneqq \operatorname{inf}_{\sigma\in Sch_{\imdp}} Pr_{\imdp}^{\sigma,l}(\psi) \\ Pr_{\imdp}^{max}(\psi) &\coloneqq \operatorname{sup}_{\sigma\in Sch_{\imdp}} Pr_{\imdp}^{\sigma,u}(\psi)
\end{split}
\end{equation}
}

The corresponding definition of an LTL formula $\psi$ holding on an MDP $\mdp$ is:
\begin{equation}
\begin{split}
    Pr_{\mdp}^{min}(\psi) &\coloneqq \operatorname{inf}_{\sigma\in Sch_{\mdp}} Pr_{\mdp}^{\sigma}(\psi) \\
    Pr_{\mdp}^{max}(\psi) &\coloneqq \operatorname{sup}_{\sigma\in Sch_{\mdp}} Pr_{\mdp}^{\sigma}(\psi)
\end{split}
\end{equation}

\section{Problem: Conservative Verification}\label{sec:problem}



\begin{figure*}[htb!]
    \centering
    \includegraphics[width=\textwidth]{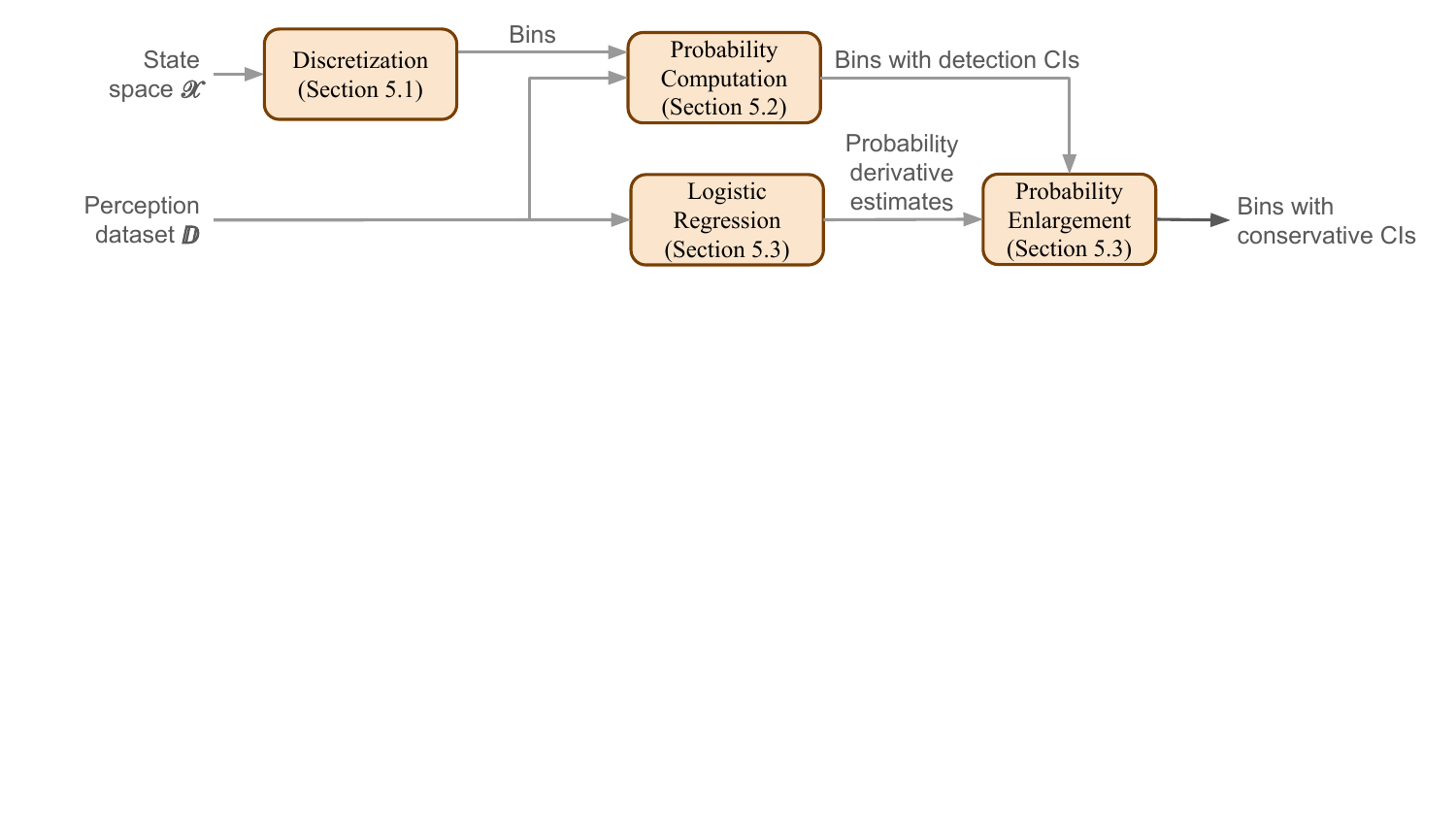}
    \vspace{-70mm}
    \caption{Overview of our approach to conservative perception modeling.}
    \label{fig:approach}
\end{figure*}

This paper addresses the problem of constructing perception models from data to enable conservative system verification. We call such perception models ``conservative perception models''. An IMDP is a \textit{conservative perception model} if it is implemented by the MDP of the true perception, per~\Cref{def:consPerOurs}. 
Consider the following \edit{discrete-time} closed-loop dynamical system denoted \sys with a \edit{binary} classifier or object detector for perception:
\begin{align}
\text{\sys}\begin{cases}
\begin{split}
    x(t+1)&=f(x(t),u(t))),\\
    y(t)&=g(x(t),v(t)), \\
    z(t)&=h(y(t)), \\
    u(t)&=c(z(t)),
\end{split}
\end{cases}
\label{eq:system}
\end{align}
where $x(t) \in \mathcal{X} \subset \mathbb{R}^n$ is the system state (with bounded $\mathcal{X}$); $y(t) \in \mathbb{R}^p$ are the observations; $z(t) \in \{0,1 \}$ is the output of a perception system\footnote{For clarity of presentation, we assume the perception system detects the presence of a single obstacle in the environment. However, our approach can handle systems with multiple types of obstacles and $z(t)$ from a fixed finite set.}, $u(t)\in\mathcal{U}\subset\mathbb{R}^m$ is the control output, the functions $f:\mathbb{R}^n\times\mathbb{R}^m\to\mathbb{R}^n$, $g:\mathbb{R}^n\times\mathbb{R}^v\to\mathbb{R}^p$, $h:\mathbb{R}^p \to \{0,1 \}$ describe the system dynamics, sensor maps, and perception maps; the function $c:\{0,1 \}\to\mathbb{R}^m$ is a stateless controller; and $v(t)\in D_v\subseteq\mathbb{R}^\mathcal{v}$ describes \editAllerton{(potentially unbounded)} perception noise. The $v(t)$ noise represents inexact sensing, such as a LiDAR scanner giving noisy distance estimates. \editAllerton{Due to the noise,~\Cref{eq:system} is a stochastic dynamical system.}

Along with \sys, we are given an LTL safety property $\psi$ to verify. The verification's goal is to compute a conservative estimate (i.e., a lower bound) of the probability that \sys satisfies $\psi$. To do this, \edit{our objective is to construct} an abstract and analyzable model of \sys, ideally as an MDP or IMDP, which we denote as \mdl. It also must be a \emph{conservative system model}, defined as follows.  
\begin{definition}[Conservative system model]
    Given system \sys as in \Cref{eq:system}, a model MDP or IMDP \mdl, and safety property $\psi$, we say the \mdl is a \emph{conservative system model} with respect to \sys if 
    \begin{equation}
        P^{min}_{\mdl}(\psi) \leq P_{\sys}(\psi)
    \end{equation}
    \label{def:conservative_system}
\end{definition}
\vspace{-4mm}

\noindent
\edit{\textbf{Model structure.}} 
When constructing \mdl, it is typical to split it into two components: a \emph{controller-plant model} \mcp (of $f$ and $c$ in Eq.~\ref{eq:system}) and a \emph{perception model} \mper (of $g$ and $h$ in Eq.~\ref{eq:system}). This decomposition is used because the behaviors in \mcp are often (non-)deterministic and/or straightforward to model using first principles (e.g., using the controller equations and system dynamics). \edit{On the other hand,} the behaviors in \mper are typically modeled probabilistically due to their complexity (e.g., large neural networks) and dependency on uncertain physical environments.

Therefore, the controller-plant model \mcp can be constructed as a conservative non-probabilistic MDP (i.e., one whose behavior contains that of true system \sys) with a standard discrete abstraction technique (for example, in~\cite{cleaveland_monotonic_2022}). However, constructing a conservative \mper, either as an MDP or IMDP, is an open problem. Therefore, this paper focuses on synthesizing \mper in a conservative manner using labeled perception data (i.e., the perception outputs $y$ paired with the true underlying states $x$) from the system \sys. In~\Cref{sec:guarantee}, we will show that our \mdl is a conservative system model, as per \Cref{def:conservative_system}. 

With \mcp and \mper in hand, we can form \mdl via parallel composition \edit{that synchronizes perception actions (e.g., detection/non-detection)}:
\begin{equation}
    \mdl = \mcp~||~\mper
\end{equation}

Before moving on, we need to define the general structure of \mcp and give a formal definition for what it means for \mcp to be a \emph{conservative controller/plant abstraction}:
\begin{definition}[Controller-Plant Model]\label{def:mcp}
Our \emph{controller-plant model} \mcp is a non-probabilistic MDP with $\mcp =  \{S_{\mcp},s^0_{\mcp},\alpha_{\mcp},\delta_{\mcp},L_{\mcp}\}$. The (discrete) state space $S$ is formed by partitioning the (continuous) state space $\mathcal{X}$ of \sys into a finite set of mutually exclusive regions, so every $s \in S_{\mcp}$ has a corresponding region $X \subseteq \mathcal{X}$ (we use $x \in s$ to denote this correspondence). The initial state $s^0_{\mcp}$ is the state in $S_{\mcp}$ which contains the initial state $x_0$ of \sys in its set of corresponding states. The \edit{actions $\alpha_{\mcp} = \alpha_{per} \times \alpha_{reach}$, where $\alpha_{per}=\{ 0,1\}$ (which we refer to as ``perception actions'') are the outputs of the perception component $h$ and $\alpha_{reach}$ (which we refer to as ``reachability actions'') are the actions that correspond to the non-determinism required to conservatively model the continuous system dynamics $f$ and controller $u$ with a discrete abstraction.\footnote{For a visualization of this phenomenon, see Figure 2 in \cite{cleaveland_monotonic_2022}.}} The transition relation $\delta_{\mcp}$ maps states \edit{and a pair (perception action,  reachability actions)} to the next state in 
$S_{\mcp}$. Finally, \edit{$L_{\mcp} : S \rightarrow 2^{AP}$, were $AP$ is a set of atomic propositions. } \toremove{if state $s$ corresponds to an unsafe state in \sys (based on our safety property $\psi$) and $L_{\mcp}(s)=\emptyset$ otherwise.} 
\end{definition}

\begin{definition}[Conservative Controller-Plant Model]\label{def:McpCons}
    $\mcp = \{S_{\mcp},s_0,\alpha_{\mcp},\delta_{\mcp},L_{\mcp} \}$ is a \emph{conservative controller/plant abstraction} if for every state $s \in S_{\mcp}$ and \edit{perception action $a_{per}\in\alpha_{per}$}, the set of next states \edit{
    $S' = \{ s' \in S \mid \exists a_{reach} \in \alpha_{reach} \text{ s.t. } \delta_{\mcp}(s,\{a_{per},a_{reach}\},s')=1 \}$} in \mcp contains every state in \sys that can be mapped to by the dynamics of \sys. Formally: \editAllerton{ 
    \begin{align*} \forall s \in S_{\mcp}, \forall a_{per} \in \alpha_{per}, \forall x \in s, \exists s' \in S, \exists a_{reach} \in \alpha_{reach} \\
     \text{such that } \delta_{\mcp}(s,\{ a_{per}, a_{reach}\},s')=1 \text{ and }  f(x,u(a)) \in s'
    \end{align*}
    }
\end{definition}

\begin{remark}
    The mismatch between the discrete abstractions and continuous system behaviors is a challenge that existing approaches fail to consider when constructing \mper abstractions.
\end{remark}


With the above in mind, this paper solves this problem:
\begin{problem}[Perception Abstraction Construction for Conservative Verification] \label{prop:ourProblem}
    Given the closed-loop dynamical system \sys defined in~\Cref{eq:system}, a safety property $\psi$, perception data $\mathcal{D}_{\text{sys}}=\{ x_i,z_i \}_{i=1,\hdots,N}$ where $x_i \in \mathbb{R}^n$ is the state of \sys and $z_i \in \{0,1\}$ is the output of the perception component, and a conservative model \mcp of the control and dynamics portion of \sys, synthesize a \editAllerton{conservative} model \mper of the perception component(s) of \sys such that
    \begin{equation}
        P^{min}_{\mcp || \mper}(\psi) \leq P_{\sys}(\psi)
    \end{equation}
\end{problem}

\section{Conservative Approach to Perception Modeling}
\label{sec:approach}

We now describe how to construct conservative models of \mper using IMDPs to solve Problem~\ref{prop:ourProblem}. Recall that we assume we are given a dataset of system states and perception outputs, denoted as $\mathcal{D}=\{ x_i,z_i \}_{ i=1,\hdots,N}$. With this dataset in hand, our modeling approach consists of three steps overviewed in \Cref{fig:approach}: 
\begin{itemize}
    \item[1. ] Discretize the continuous state space $\mathcal{X}$ to obtain state set $S$, since MDP and IMDP models require discrete state spaces.
    \item[2. ] For each discrete state, compute an interval of detection probabilities using $\mathcal{D}$ and binomial confidence intervals \cite{fleiss_statistical_2004}. This is a novel step where we differ from existing perception models~\cite{badithela_evaluation_2023,cleaveland_conservative_2023}.
    \item[3. ] Enlarge the detection probability ranges from Step 2 to account for the change in detection probability as the system's (continuous) state $x\in\mathcal{X}$ changes. This is required to fully bridge the gap between the continuous parts of~\Cref{eq:system} and our discrete models.
\end{itemize}

\subsection{Step 1: State Space Discretization}
\label{sec:stateSpaceBinning}
The first step is to partition the state space $\mathcal{X}$ of the system into a finite, discrete set of bins, which we denote as $S$. Each bin $s\in S$ corresponds to region $\mathcal{X}_i \subseteq \mathcal{X}$. With a slight abuse of notation, we say that $x \in s$ when $x \in \mathcal{X}_i$. \edit{Note that these bins are specifically for building $M_{per}$ and do \textbf{not} have to match those of $M_{cpl}$: we can use non-deterministic transitions to account for states in $M_{cpl}$ that overlap multiple states in $M_{per}$, and vice versa.}

In terms of performing the binning in practice, we see two main ways of doing this. First, one can divide the state space into a set of equally spaced hyper-rectangles (or intervals if the state space is 1D, rectangles in 2D, ...)\footnote{This is the approach taken by the previous works~\cite{badithela_evaluation_2023,cleaveland_conservative_2023}.}. In this case, the size of the hyper-rectangle serves as a hyper-parameter of the approach. The second way is to first select the desired number of bins \edit{for each dimension of the state space $\mathcal{X}$}, which we denote as \edit{$\mathcal{b_1},\hdots,\mathcal{b_n}$}, and then divide the state space such that each bin has an equal number of datapoints in it. For example, if one has a 1D system and $l$ datapoints (so $|\mathcal{D}| = l$), then one would sort the points $\{x_i,z_i\} \in \mathcal{D}$ by the $x_i$ values and assign each chunk of $\frac{\mathcal{l}}{n}$ to its own bin.

\looseness=-1
The choice of a binning method is subject to important trade-offs. The main and well-known one is between the granularity of binning and the accuracy/cost of verification. The main hyper-parameter is the number of bins used to partition $\mathcal{X}$. Using a lot of bins to create a fine partition will lead to an \mper with a larger state space, which is more expensive to model check. On the plus side, a model with more bins bridges the discrete-continuous gap better, since the behaviors of the perception component change continuously over the state space. 

We point out two other, subtler trade-offs at play that our discretization-based approach identifies. The existing perception abstraction approaches~\cite{badithela_evaluation_2023,cleaveland_conservative_2023} fail to account for these trade-offs, since they ignore the inherent uncertainty present in \mper --- and so they fail to make it conservative. The first tradeoff is that using more bins means that there will be fewer datapoints per bin. This will lead to larger probability intervals for each bin in Step 2 since there are fewer datapoints fed into the binomial confidence intervals. These larger detection probability intervals in turn will lead to a larger interval on the estimated safety probability $\hat{P}(\psi)$. This trade-off captures the intuition that having fewer datapoints to compute the transition probabilities of \mper generally results in a less precise model. 

Another subtle trade-off relates the bin sizes to the detection probability enlargement in Step 3. Using fewer, larger bins means that the detection probability intervals computed in Step 2 will have to be inflated by a larger amount to account for the perception system's performance changing over the continuous state space $\mathcal{X}$. This trade-off captures the intuition that a model with larger (discrete) bins will generally fail to capture the changes in the performance of the perception component over the continuous state space $\mathcal{X}$, which leads to a less accurate model. These two subtle tradeoffs are shown in~\Cref{fig:binWidthTradeoff}. 

\begin{figure}
    \centering
    \includegraphics[width=0.8\linewidth]{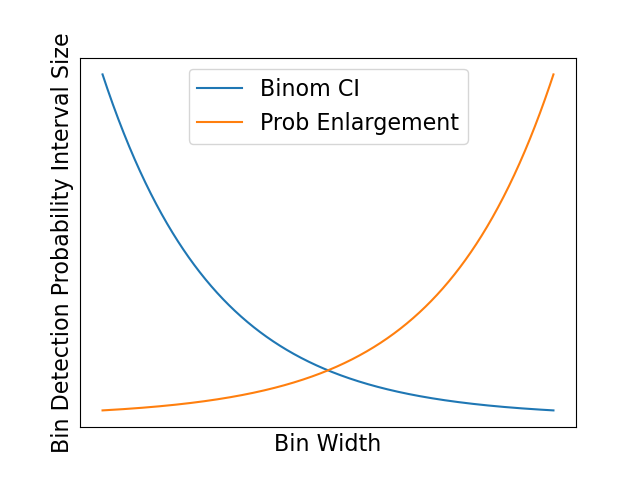}
    \caption{Larger discretization bins lead to smaller binomial confidence intervals (blue) but necessitate more probability enlargement (orange). 
    }
    \label{fig:binWidthTradeoff}
\end{figure}

\subsection{Step 2: Bin Probability Computation}
\label{sec:binProbComp}

The second step of our approach uses the datapoints in $\mathcal{D}$ to estimate the detection probabilities of the perception component $h$ within each bin $s \in S$. First, we denote the datapoints from $\mathcal{D}$ that lie in $s$ as 
\begin{equation}
    \mathcal{D}_{s} = \{ (x,z) \in \mathcal{D} ~|~ x \in s  \}
\end{equation}

Just as in prior works~\cite{badithela_evaluation_2023,cleaveland_conservative_2023}, we can then compute the empirical detection probability of the datapoints in $\mathcal{D}_{s}$:
\begin{align}
    \label{eq:binDetProb}
    \hat{p}_{s} = \frac{\sum_{(x,z) \in \mathcal{D}_{s}} \mathbbm{1}_{z = 1} }{|\mathcal{D}_{s}|}
\end{align}

At this step, the previous approaches in the literature would put these estimates as transition probabilities in \mper. However, these point estimates of the detection probability are not reliable: they are themselves random values subject to data collection noise, not equal to the true underlying probabilities. Thus, the conservatism of \mper is at risk. 

To alleviate this issue, we replace the point estimate $\hat{p}_{s}$ with a \emph{confidence interval (CI)}, denoted as $[\underline{\hat{p}}_{s},\bar{\hat{p}}_{s}]$ for bin $s$.\footnote{Using CIs to construct IMDPs has been explored in the areas of synthesis \cite{jiang_safe_2022} and active learning \cite{alasmari_quantitative_2022} (which inspires our approach), but not verification.}  There are several methods to obtain a confidence interval for the detection probability of the perception model within the bin, also known as a binomial confidence interval. We employ the Clopper-Pearson confidence interval because its coverage level is never less than the desired nominal coverage~\cite{clopper_use_1934}.
When constructing a confidence interval, one needs to select a confidence level $1-\alpha$. Using a larger value of $1-\alpha$ leads to larger probability intervals, but the proportion of bins in which the confidence interval contains the true detection probability will, in turn, be larger. 

We argue that the choice of $1-\alpha$ should be determined by the desired confidence in the model checking result for the overall system \sys. If one wants a confidence of $1-\alpha_{MC}$ that the model checking results are conservative (i.e. ~\Cref{def:conservative_system} holds with probability $1-\alpha_{MC}$) and there are $|S|$ states in \mper, then one should set $1-\alpha=1-\alpha_{MC}/|S|$. Intuitively, this formula captures the fact that each bin's confidence interval has \edit{an} $\alpha_{MC}/|S|$ probability of \edit{not} containing the true detection probability, so by applying a union bound \cite{boole_mathematical_1847}, the probability of every bin's confidence interval containing the true detection probability is lower bounded by
\begin{align}
    1- \sum_{s \in S} \frac{\alpha_{MC}}{|S|} = 1-\alpha_{MC}
\end{align}

We will use this union-bound argument when proving the main theoretical result of this paper in \Cref{thm:thm1}.

\subsection{Step 3: Probability Enlargement}
\label{sec:ProbEnlargement}

As the third and final step, we expand the confidence intervals from Step 2 to account for the fact that the perception model's detection probability can change within each bin. Intuitively, the confidence intervals we computed in the previous step only account for the \textit{average} detection probability within each bin, so we cannot say if it covers \textit{all} the detection probabilities for every system state in the bin. Thus, the argument for the conservatism of \mper would be incomplete.

To account for these potential changes in intra-bin detection probabilities, we need a way of inflating the bin confidence intervals based on how much the underlying detection probability changes within each bin. We perform this inflation by estimating the derivative of the detection probability with respect to the system state using a \textit{surrogate model}. 

As the surrogate model, we fit a logistic regression model to the entire dataset $\mathcal{D}$. This model takes as input the system state $x$ and outputs the estimated probability that the perception component will detect the object, or $\hat{P}(z=1~|~x)$. We denote the function as $LR: \mathcal{X} \rightarrow [0,1]$. Then, we compute how much the estimated detection probability changes within each bin. For bin $s$, we compute this \textit{probability enlargement factor} as 
\begin{align} \label{eq:probEnlargeDelta}
    \delta^{p}_{s} := \left(\max_{x\in s} \; LR(x) \right) - \left(\min_{x\in s} \; LR(x) \right)
\end{align}

Now, the updated detection probability interval for bin $s$ is 
\begin{align}\label{eq:binDetIntOurs}
    [\underline{\hat{p}}_{s}-\delta^{p}_{s},\bar{\hat{p}}_{s}+\delta^{p}_{s}]
\end{align}

\begin{remark} If we fit a surrogate model to detection probabilities, why not use it to construct \mper? 
    We attempted to directly use the logistic regression, along with its associated confidence interval, to estimate the range of true detection probabilities within the bins. However, we found that the logistic regression confidence interval did not empirically achieve the coverage that we wanted. That is because the logistic regression CIs can be viewed as the CIs for the \textit{parameters} of how to map data onto the space of logistic models --- not for the data itself. 
    We demonstrate this in the experiments in \Cref{sec:expCARLA} and particularly \Cref{fig:carlaDetProbs}.
\end{remark}

\section{Conservative Verification Guarantee}
\label{sec:guarantee}

In this section, we state the main theoretical guarantee of this work. Specifically, we show that given the inputs to Problem~\ref{prop:ourProblem} and an \mper constructed using Steps 1--3 from \Cref{sec:approach}, verifying the resulting composed model \mdl will give a conservative estimate of the safety of the true system \sys:
\begin{align}
    P_{\sys}(\psi) \in [Pr_{\mdl}^{min}(\psi),Pr_{\mdl}^{max}(\psi)]
\end{align}

Before the theorem and its proof, we define some notation and a lemma. The proofs of the lemma and theorem can be found \edit{in the Appendix}.

\begin{definition}[State-level true detection probability]
    Let $P_{x}(h)$ for $x \in \mathcal{X}$ denote the (ground-truth) \emph{detection probability} of the system's perception component $h$ performing a positive detection when the system \sys is in state $x$.
\end{definition}

\begin{definition}[Bin-level true detection probability range]
    Let bin $s_i \in S$ denote a discrete perception bin with corresponding state space region $\mathcal{X}_i \subseteq \mathcal{X}$, as described in \Cref{sec:stateSpaceBinning}. Now let $P_{s_i}(h)$ denote the \emph{range of the detection probabilities} of perception $h$  over the states $x \in s_i$:
    \begin{align}
        P_{s_i}(h) = \left[\min_{x\in s_i} P_{x}(h), \max_{x \in s_i} P_{x}(h)\right]
    \end{align}
\end{definition}

\begin{lemma}[Bin-level detection probability conservatism]\label{lem:binConservatism}
    Consider the same inputs as in Problem~\ref{prop:ourProblem}. Let $[\underline{\hat{p}}_{s_{i}},\bar{\hat{p}}_{s_{i}}]$ denote the $1-\frac{\alpha_{MC}}{|S|}$ confidence interval of the perception component detection probability in bin $s_i$ (as described in \Cref{sec:binProbComp}). Assume\footnote{\edit{In practice, this assumption means that the true detection probabilities can be instantiated by some parameter values of Step 3's surrogate model, e.g., the logistic regression. In our case studies, this assumption has been largely upheld. Its theoretical investigation is left for future work.}} the probability enlargement factor $\delta_{s_i}$ as computed in \Cref{eq:probEnlargeDelta} upper-bounds the size of $P_{s_i}(h)$: 
    \begin{align}
        \delta_{s_i} \geq \left|\max_{x \in s_i} P_x(h) - \min_{x \in s_i} P_x(h) \right|
    \end{align}
    Then, with probability $1-\frac{\alpha_{MC}}{|S|}$, the interval from \Cref{eq:binDetIntOurs} contains the true detection probability interval $P_{\mper}(s_i)$. Precisely:
    \begin{align*}
        \underline{\hat{p}}_{s_{i}}-\delta_{s_{i}} \leq \min_{x\in s_i} P_{x}(h) \\
        \bar{\hat{p}}_{s_{i}}+\delta_{s_{i}} \geq \max_{x \in s_i} P_{x}(h)
    \end{align*}
    \editAllerton{holds with probability $1-\frac{\alpha_{MC}}{|S|}$.}
\end{lemma}

\begin{theorem}[Model-level safety probability conservatism]\label{thm:thm1}
    Consider the same inputs as in Problem~\ref{prop:ourProblem}, \edit{the assumption of Lemma~\ref{lem:binConservatism}}, and perception model \mper with the transition probabilities from \Cref{eq:binDetIntOurs} using confidence level $1-\frac{\alpha_{MC}}{|S|}$. Letting $\mdl=\mcp~||~\mper$, it holds that
    \editAllerton{\begin{align}
        P\big(P_{\sys}(\psi) \in [Pr_{\mdl}^{min}(\psi),Pr_{\mdl}^{max}(\psi)]\big) \geq 1-\alpha_{MC}
    \end{align}
    }
\end{theorem}

\section{Case Studies}
\label{sec:caseStudies}

Our experimental evaluation has two goals: explore how the bin size hyperparameter and how the probability enlargement step affect the conservatism of our approach on a fully synthetic case study. We use an automatic emergency braking system (AEBS) on an autonomous car to evaluate our approach. Additionally, we apply our approach to an autonomous car example in the CARLA simulator \cite{dosovitskiy_carla_2017}. \edit{For IMDP verification, we use PRISM model checker 4.8.1 with native IMDP support.}

\subsection{System: Emergency Braking with Object Detection}

Consider an autonomous car approaching a stationary obstacle. The car is equipped with a controller that issues braking commands on detection of the obstacle and a perception LEC that detects obstacles. The safety goal of the car is to fully stop before hitting the obstacle. We use the following car dynamics, controller, and perception.

A discrete kinematics model with time step \timestep represents the velocity ($v$) and position ($d$, same as the distance to the obstacle) of the car at time $t$:
\begin{equation} \label{eq:aebs-dyn}
    d[t] = d[t-\timestep] - \timestep \times v[t-\timestep], \quad
    v[t] = v[t-\timestep] - \timestep \times b[t-\timestep],
\end{equation}
where $b[t]$ is the braking command (a.k.a. the ``braking power'', BP) at time $t$.

We use an \textit{Advanced Emergency Braking System (AEBS)} \cite{lee_aebs_2011} which uses two metrics to determine the BP: time to collide ($TTC$) and warning index ($WI$). The $TTC$ is the amount of time until a collision if the current velocity is maintained. The $WI$ represents how safe the car would be in the hands of a human driver (positive is safe, negative is unsafe).
\begin{equation}
    TTC = \frac{d}{v},\quad WI = \frac{d-d_{br}}{ v  T_{h}}, \quad d_{br} = v T_{s}+ \frac{u v^2}{2a_{max}}, 
\end{equation}
where $d_{br}$ is the braking-critical distance, $T_{s}$ is the system response delay (negligible), $T_{h}$ is the average driver reaction time (set to 2s), $u$ is the friction scaling coefficient (taken as 1), and $a_{max}$ is the maximum deceleration of the car.

Upon detecting the obstacle, the AEBS chooses one of three BPs: no braking (BP of 0), light braking ($B_1$), and maximum braking ($B_2=a_{max}$). If the obstacle is not detected, no braking occurs. The  BP value, $b$, is determined by $WI$ and $TTC$ crossing either none, one, or both of the fixed thresholds $C_1$ and $C_2$:
\begin{equation}
\begin{split}
   &WI >C_1 \land TTC > C_2 \implies  b = 0 \\
    &(WI \leq C_1 \land  TTC > C_2) \lor (WI > C_1 \land  TTC \leq C_2) \\ 
    &\quad\quad\implies b = B_1  \\
    &WI \leq C_1 \land  TTC \leq C_2 \implies b = B_2
\end{split}
\end{equation}

To detect the obstacle, the car uses a classification variant of the deep neural network YoloNetv11~\cite{yolo11_ultralytics},
as it can run at high frequencies. We say a low-level detection of the obstacle occurs when Yolo detects the obstacle.

The safety property \oursafeprop is the absence of a collision, specified in LTL as 
\begin{equation}
    \oursafeprop :=  \square \; \left ( d > \carlen \right ),
\label{eq:aebs_safety_prop}
\end{equation}
where \carlen is the minimum allowed distance to the obstacle and is taken to be 5m. Note that in this model the car eventually stops or collides, so for any initial condition there is an upper bound on the number of time steps.

\subsection{Comparative Abstraction Methods}

We compare our approach to a few baseline/ablation methods of creating \mper summarized in Table~\ref{tab:perceptionApproachesEXps}: a baseline approach~\cite{badithela_evaluation_2023} that doesn't use confidence intervals, our approach with no probability enlargement (this is our approach, but skipping~\Cref{sec:ProbEnlargement}), and an approach that uses confidence intervals from the logistic regression we employ in~\Cref{sec:ProbEnlargement} to compute the bin probability intervals \editAllerton{(this is our approach, but skipping~\Cref{sec:binProbComp})}, and the approach that uses the ground truth perception detections probabilities (for the synthetic case study only). We expect our approach to be more conservative than all of these due to its guarantees.


\begin{table*}[]
    \centering
    \begin{tabular}{|c|c|}
        \hline \textbf{Model Symbol} & \textbf{Model Description} \\
        \hline \modelNoCI & Baseline using $\hat{p}_{s}$ from~\Cref{eq:binDetProb} as detection probabilities~\cite{badithela_evaluation_2023} \\
        \hline \modelGTPer & Uses the ground-truth detection probabilities \\
        \hline \modelLogRegCI & Approach using the logistic regression confidence intervals \\
        \hline \modelNoDeriv & Our approach with no probability enlargement \\
        \hline \modelOurs & Our approach: enlargement with derivatives from logistic regression   \\ 
        \hline        
    \end{tabular}
    \vspace{1mm}
    \caption{Perception abstraction approaches compared in our experiments. }
    \label{tab:perceptionApproachesEXps}
\end{table*}

\subsection{Synthetic perception case study}
\label{sec:expSimulated}

For the first case study, we replaced YoloNetv11 with a simulated stochastic perception model. The purpose is to experiment with the fully known ground truth. This true perception model uses a dynamically weighted coin flip to determine detections and no detections, where the probability of detection depends on the distance from the car to the obstacle. This dependency is modeled with a logistic function, computing the detection probability based on the obstacle distance:
\begin{align}
    p(x) = \frac{1}{1+e^{-k(x-x_0)}}
\end{align}
where $k=-0.1$ and $x_0=35$. A plot of this function over the range of distances we considered for the problem is shown in~\Cref{fig:perceptionModel}.

\begin{figure}
    \centering
    \includegraphics[width=0.8\linewidth]{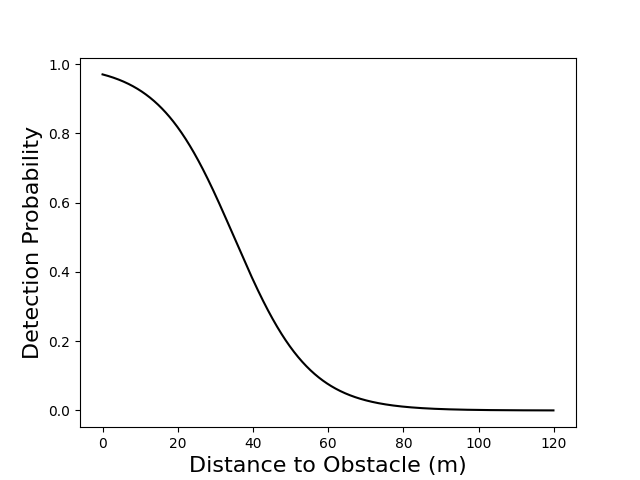}
    \caption{Detection probabilities for the synthetic perception model.}
    \label{fig:perceptionModel}
    \vspace{-3mm}
\end{figure}

We investigate how the bin width affects the safety estimates of our method. Additionally, we vary the probability enlargement by a constant multiple to investigate how that aspect of our approach contributes to the safety estimates. \edit{For brevity, we omit the reporting of binning algorithms that determine bin boundaries, which produced little to no effect.}

\subsubsection{Bin Width Effect: }
For the first experiment, we investigate how the width of the bins affects the conservatism of our perception abstractions\footnote{\editAllerton{We note that the bin size, both for \mper and \mcp also affects the runtime of the model checking. See e.g. \cite{cleaveland_monotonic_2022} for an investigation of that. }}. To do this, we vary the width of the bins to see how that affects the safety probabilities of the various perception abstraction generation methods: \modelOurs, \modelNoDeriv,\modelNoCI, \modelLogRegCI, and \modelGTPer. The results for initial distance 50$m$ and initial speed $20$m/s are shown in~\Cref{fig:simulatedAEBSResultsBinSweep}. First, note that the range of safety probabilities from \modelOurs (green) contains the empirical safety probability across the range of bin sizes, highlighting the conservatism of our approach. Second, the safety ranges from \modelNoDeriv (red) contain the empirical safety probability only for very small bin sizes, but not for moderate to large ones. This highlights the necessity of the probability enlargement from~\Cref{sec:ProbEnlargement}. Additionally, the gap between \modelOurs and \modelNoDeriv grows as the bin size increases, which is as expected because the effect of the probability enlargement grows as the bin size increases. Third, note that the behavior of \modelNoCI (yellow) resembles that of \modelNoDeriv, just with smaller safety probability ranges. Note that \modelNoCI still results in a range of safety probabilities (rather than a point estimate) due to non-determinism in \mcp. Finally, \modelLogRegCI (cyan) produces nearly identical results to \modelGTPer. This is because the detection probability is exactly determined by a logistic function, (see~\Cref{fig:perceptionModel}), so the logistic regression is almost perception accurate. Our CARLA case study shows that \modelLogRegCI underperforms when it does not match the model class of the true perception model. \edit{While the safety probability bounds from \modelOurs are quite conservative for larger bin widths, we note that the bounds are conservative for \textit{all} bin widths, while the other approaches produce probability bounds that do not contain the true safety probability. The reduction of this conservatism is left for future work.}

Overall, this experiment highlights that our approach produces conservative models regardless of bin size --- while the existing approaches do not. \footnote{\editAllerton{Because our method is conservative for any bin size, the user can test a range of bin sizes and use whichever one gives the tightest safety probability interval.}}


\begin{figure*}
    \centering
    \begin{subfigure}[t]{0.555\textwidth}
        \centering
        \includegraphics[width=\textwidth]{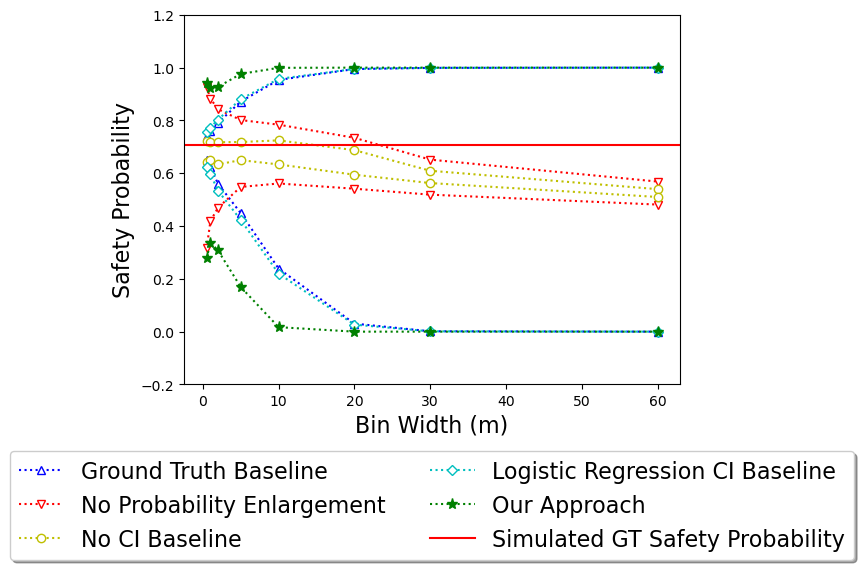}
        \caption{Safety probability estimates from several approaches over a range of bin widths.}
        \label{fig:simulatedAEBSResultsBinSweep}
    \end{subfigure}%
    ~ 
    \begin{subfigure}[t]{0.435\textwidth}
        \centering
        \includegraphics[width=\textwidth]{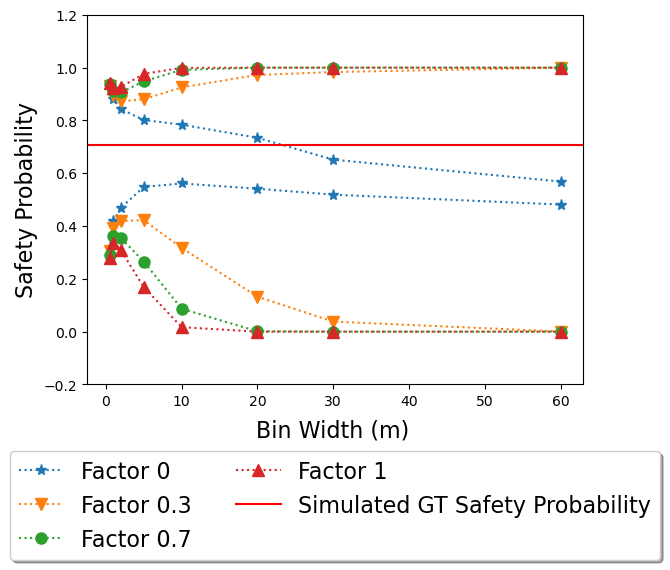}
        \caption{Our safety probability estimates over a range of bin widths and enlargement factors.}
        \label{fig:simulatedAEBSResultsEnlargeMultSeep}
    \end{subfigure}
    \caption{Results of the synthetic perception case study.}
    \vspace{-3mm}
\end{figure*}

\subsubsection{Probability Enlargement Factor: }

This experiment aims to show that our approach can flexibly incorporate the user's risk preference. To this end, we vary how much we enlarge the bin detection probability intervals by a multiplicative factor $\probEnlargeFactor\in [0,1]$. So the detection probability interval for bin $s_i$ from \Cref{eq:binDetIntOurs} now becomes
\begin{align}
    [\underline{\hat{p}}_{s_{i}}-\probEnlargeFactor \delta^{p}_{s_{i}},\bar{\hat{p}}_{s_{i}}+\probEnlargeFactor \delta^{p}_{s_{i}}]
\end{align}

\looseness=-1
Using the same experimental conditions as before, we tested our approach using $\probEnlargeFactor\in \{0,0.3,0.7,1\}$, where $\probEnlargeFactor=1$ corresponds to our full approach \modelOurs, and $\probEnlargeFactor=0$ corresponds to the no probability enlargement \modelNoDeriv. The results of this experiment are shown in~\Cref{fig:simulatedAEBSResultsEnlargeMultSeep}. First, note that the safety intervals are monotonically increasing as the factor increases. That is because the larger probability enlargement factor increases the detection probability intervals in \mper. For \editAllerton{$\probEnlargeFactor \in \{0.3,0.7 \}$}, the approach remains empirically conservative. We conclude that bridging the gap between the probability enlargement needed for theoretical and empirical conservatism is an interesting avenue for future research.


\subsection{CARLA case study}
\label{sec:expCARLA}

\looseness=-1
Here, we evaluate the feasibility of our approach for a realistic object detection LEC in a high-fidelity simulation. First, we set up the AEBS controller in the autonomous driving simulator CARLA~\cite{dosovitskiy_carla_2017}, using a red Toyota Prius as the obstacle and a classification variant of Yolov11 as the object detector. Then we performed two rounds of data collection in CARLA: (i) constant-$v$ runs to gather image-distance data to build our perception abstractions, (ii) AEBS-controlled runs to approximate the true safety probability of the AEBS controller.

For the constant-$v$ runs, we ran $500$ trials with a speed of $6$m/s and initial distance around $210$m, which generated $181000$ perception datapoints. \Cref{fig:carlaDetProbs} shows 5m-wide binned estimates of the YoloNet detection probabilities from the perception data. For the AEBS-controlled runs, we ran $500$ trials with an initial distance $50$m and an initial speed $20$m/s. This resulted in $330$ crashes, which yields a $[0.62, 0.70]$ 95\%-CI for the true collision chance.


\begin{figure*}
    \centering
    \begin{subfigure}[t]{0.47\textwidth}
        \centering
        \includegraphics[width=0.9\textwidth]{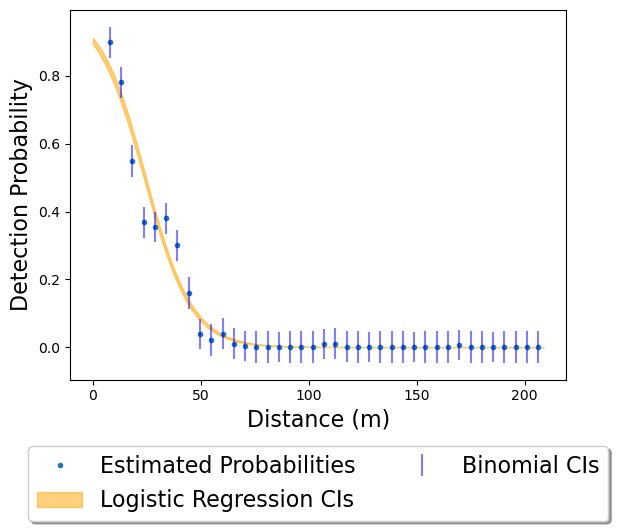}
        \caption{Detection probabilities vs. distance to obstacle for our Yolov11  LEC.}
        \label{fig:carlaDetProbs}
    \end{subfigure}%
    ~~~
    \hspace{1mm}
    \begin{subfigure}[t]{0.47\textwidth}
        \centering
        \includegraphics[width=\textwidth]{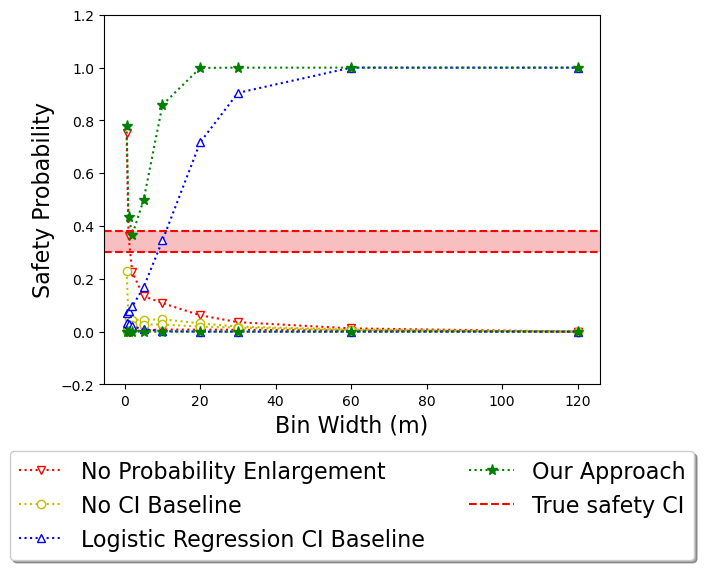}
        \caption{Safety probability estimates over a range of bin widths.}
        \label{fig:CARLAAEBSResultsBinSweep}
    \end{subfigure}
    \caption{Results of the CARLA case study.}
    \vspace{-3mm}
\end{figure*}

We used the perception data to construct models using our approach \modelOurs, our approach with no probability enlargement \modelNoDeriv, the approach from \cite{badithela_evaluation_2023} with no confidence intervals \modelNoCI, and the approach using the logistic regression confidence intervals \modelLogRegCI. Note that we cannot construct \modelGTPer since we don't have access to the ground truth detection probabilities.

We first vary the bin width parameter used to form the models to determine how it affects the estimated safety probability intervals. The results are shown in \Cref{fig:CARLAAEBSResultsBinSweep}. Note that \edit{for every bin width} the safety probability intervals from our approach contain the confidence interval for the true underlying safety probability. Also, the probabilities from \modelNoDeriv are only empirically valid for small bins, as we witnessed in~\Cref{sec:expSimulated}. Additionally, the probabilities from \modelNoCI never contain the true safety probability. Finally, the safety probabilities from \modelLogRegCI are not valid for small bin widths. This highlights the drawbacks of solely relying on logistic regression to estimate detection probabilities in practice.


\section{Conclusion and Future Work}
\label{sec:discussion}

\looseness=-1
This paper has presented a novel method for constructing provably conservative probabilistic abstractions of perception and sensing components in closed-loop autonomous systems. 
Our abstraction method combines (i) binomial confidence intervals to account for the inherent uncertainty of data-driven estimation and (ii) logistic regression to account for the perception behaviors changing continuously in the system's state space. We prove that when our perception/sensing abstractions are composed with conservative control and dynamics models, the resulting model will lead to conservative verification with high probability. We demonstrate the conservatism of our approach in case studies of an AEBS system with both synthetic perception models and in CARLA using a realistic object detector. 

Future work can focus on reducing conservatism by better estimating intra-bin changes in detection probabilities of the perception component. For example, more sophisticated surrogate models for probability. \editAllerton{Additionally, we will investigate more expressive \mper models, such as conditioning detections on previous \mper outputs, and explore abstractions for multi-class LECs. Finally, we will extend our approach to specifications beyond LTL safety. }

\bibliographystyle{splncs04}

\bibliography{ivan-autogenerated,literature}

\appendix

\section{Appendix}
\setcounter{lemma}{0}
\setcounter{theorem}{0}

\subsection{Proof of~\Cref{lem:binConservatism}}
\label{app:lemma1}

\begin{lemma}[Bin-level detection probability conservatism]\label{lem:binConservatism_app}
    Consider the same inputs as in Problem~\ref{prop:ourProblem}. Let $[\underline{\hat{p}}_{s_{i}},\bar{\hat{p}}_{s_{i}}]$ denote the $1-\frac{\alpha_{MC}}{|S|}$ confidence interval of the perception component detection probability in bin $s_i$ (as described in \Cref{sec:binProbComp}). Assume the probability enlargement factor $\delta_{s_i}$ as computed in \Cref{eq:probEnlargeDelta} upper-bounds the size of $P_{s_i}(h)$: 
    \begin{align}
        \delta_{s_i} \geq \left|\max_{x \in s_i} P_x(h) - \min_{x \in s_i} P_x(h) \right|
    \end{align}
    Then, with probability $1-\frac{\alpha_{MC}}{|S|}$, the interval from \Cref{eq:binDetIntOurs} contains the true detection probability interval $P_{s_i}(h)$. Precisely:
    \begin{align*}
        \underline{\hat{p}}_{s_{i}}-\delta_{s_{i}} \leq \min_{x\in s_i} P_{x}(h) \\
        \bar{\hat{p}}_{s_{i}}+\delta_{s_{i}} \geq \max_{x \in s_i} P_{x}(h)
    \end{align*}
\end{lemma}
\begin{proof}
    By the use of the $1-\frac{\alpha_{MC}}{|S|}$ confidence interval when computing $[\underline{\hat{p}}_{s_{i}},\bar{\hat{p}}_{s_{i}}]$, we know that average detection probability of \mper over bin $s_i$ is contained in $[\underline{\hat{p}}_{s_{i}},\bar{\hat{p}}_{s_{i}}]$ with probability $1-\frac{\alpha_{MC}}{|S|}$. Next, by the assumption that $\delta_{s_i}$ upper bounds the size of $P_{s_i}(h)$, it follows that 
    \begin{align*}
        \underline{\hat{p}}_{s_{i}}-\delta_{s_{i}} \leq \min_{x\in s_i} P_{x}(h) \\
        \bar{\hat{p}}_{s_{i}}+\delta_{s_{i}} \geq \max_{x \in s_i} P_{x}(h)
    \end{align*}
\end{proof}

\subsection{Proof of~\Cref{thm:thm1}}
\label{app:thm1}

\begin{theorem}[Model-level safety probability conservatism]\label{thm:thm1_app}
    Consider the same inputs as in Problem~\ref{prop:ourProblem}, \edit{the assumption of Lemma~\ref{lem:binConservatism}}, and perception model \mper with the transition probabilities from \Cref{eq:binDetIntOurs} using confidence level $1-\frac{\alpha_{MC}}{|S|}$. Letting $\mdl=\mcp || \mper$, it holds that
    \begin{align}
        P_{\sys}(\psi) \in [Pr_{\mdl}^{min}(\psi),Pr_{\mdl}^{max}(\psi)]
    \end{align}
\end{theorem}

\begin{proof}
    Our proof has two general steps. First, we show that if the perception bin detection probability interval of \mper contains the true range of detection probabilities, then \mdl will be conservative. Then we use~\Cref{lem:binConservatism} to bound the probability that the previous assumption holds.
    
    Let \sys be in state $x\in\mathcal{X}$. Its perception component $h$ either outputs a detection ($h(g(x,v)=1$) with probability $P_{x}(h)$ or a non-detection ($h(g(x,v)=0$) with probability $1-P_{x}(h)$ (recall that $v$ is an unknown, random noise variable). In the event of a detection, the next state of \sys will be $x_{1}'=f(x,c(1))$ and in the event of a non-detection, the next state of \sys will be $x_{0}'=f(x,c(0))$. This induces the distribution of the next state $x'$ of \sys. 
    
    Now consider our abstraction \mdl stating in its state $s$ that corresponds to $x$ (i.e. $x \in s$). Let's assume that the perception bin detection probability interval of \mper for state $x$ contains $P_{x}(h)$. So the probability interval that \mper sends $a_{per}=1$ to \mcp contains the true detection probability $P_{x}(h)$ and the probability interval that \mper sends $a_{per}=0$ contains the true mis-detection probability $1-P_{x}(h)$. By the assumption that \mcp is conservative (per Problem~\ref{prop:ourProblem}), if \mper sends a detection \edit{($a_{per}=1$)} then the set of next states \edit{$S_{1}' = \{ s' \in S | \exists a_{reach} \in \alpha_{reach} \text{ such that } \delta_{\mcp}(s,\{ 1,a_{reach}\},s')=1 \}$} that \mcp can transition to will contain $x_{1}'$ ($\exists s_{1}' \in S_{1}'$ such that $x_{1}' \in s_{1}'$) and if \mper sends a no detection \edit{($a_{per}=0$)} then the set of next states \edit{$S_{0}' = \{ s' \in S | \exists a_{reach} \in \alpha_{reach} \text{ such that } \delta_{\mcp}(s,\{ 0,a_{reach}\},s')=1 \}$} that \mcp can transition to will contain $x_{0}'$ ($\exists s_{0}' \in S_{0}'$ such that $x_{0}' \in s_{0}'$). So starting from state $x$, the \editAllerton{set of distributions} of the next state defined by $\mcp || \mper$ contain the distribution of the next state for \sys if the detection probability interval for every bin $s_{i}$ of \mper \editAllerton{contains $P_{s_{i}}(h)$}. 

    Now by \Cref{lem:binConservatism}, the perception bin detection probability interval of \mper for state $x$ contains $P_{x}(h)$ with probability $1-\frac{\alpha_{MC}}{|S|}$.

    Finally, noting that \mper has $|S|$ bins, the probability that every perception bin's detection probability interval contains the true range of detection probabilities is $$ 1-\sum_{b\in S} \frac{\alpha_{MC}}{|S|}  \leq 1- \alpha_{MC}$$
    
\end{proof}

\end{document}